\documentclass[graybox]{svmult}

\usepackage{amsmath,mathptmx,amssymb,graphicx}       
\usepackage{helvet}         
\usepackage{courier}        
\usepackage{type1cm}        
\usepackage{makeidx}         
\usepackage{graphicx}        
\usepackage{multicol}        
\usepackage[bottom]{footmisc}
\usepackage{graphicx} \usepackage{epstopdf}

\makeindex             

\begin{document}
	
	\title*{Convergence of solutions in a mean-field model of go-or-grow type with reservation of sites for proliferation and cell cycle delay}
	 \titlerunning{A mean field model with reservation of sites for proliferation and cell cycle delay}
	\author{Ruth E. Baker, P\'eter Boldog and Gergely R\"ost}
	\institute{Ruth E. Baker \at Mathematical Institute, University of Oxford, UK
		\and P\'eter Boldog \at Bolyai Institute, University of Szeged, Hungary 
		\and Gergely R\"ost \at Bolyai Institute, University of Szeged, Hungary \email{rost@math.u-szeged.hu} \\   Mathematical Institute, University of Oxford, UK}

\maketitle

\abstract{We consider the mean-field approximation of an individual-based model describing cell motility and proliferation, which incorporates the volume exclusion principle, the go-or-grow hypothesis and an explicit cell cycle delay. To utilise the framework of on-lattice agent-based models, we make the assumption that cells enter mitosis only if they can secure an additional site for the daughter cell, in which case they occupy two lattice sites until the completion of mitosis. The mean-field model is expressed by a system of delay differential equations and includes variables such as the number of motile cells, proliferating cells, reserved sites and empty sites. We prove the convergence of biologically feasible solutions: eventually all available space will be filled by mobile cells, after an initial phase when the proliferating cell population is increasing then diminishing. By comparing the behaviour of the mean-field model for different parameter values and initial cell distributions, we illustrate that the total cell population may follow a logistic-type growth curve, or may grow in a step-function-like fashion.}

\section{Introduction}

Cell proliferation and motility are key processes that govern cancer invasion or wound healing. The go-or-grow hypothesis postulates that proliferation and migration spatiotemporally exclude each other. This has been acknowledged, for example, for glioblastoma \cite{giese}. In general, two phenotypes that can be of particular importance to progression of aggressive cancers are `high proliferation-low migration' and `low proliferation-high migration', and the mechanisms governing this switching are of great interest in current medical research \cite{levchenko}. Here we consider a strong simplification of this phenomenon by assuming that (differently from \cite{geerle}) motile cells stop for a fixed period of time to complete cell division, upon which they immediately switch back into the migratory phenotype. We study the mathematical properties of a mean-field approximation of an individual based model describing this process, and this note complements our other ongoing works \cite{br,bbr} where we investigate in detail a range of biological hypotheses with the corresponding individual-based as well as mean-field, analytically tractable, models.

\section{The model}

Assume that agents (representing biological cells) move and proliferate on an $n$-dimensional
square lattice with length $\ell$ (in each direction), so that $K=\ell^n$ is an integer
describing the number of lattice sites. We divide our agent population into two subpopulations, motile and proliferative, with the condition that a proliferative agent has to be attached to an adjacent site which is reserved until the end of proliferation. As a result, sites can either contain a motile agent, a proliferating agent, be reserved for the daughter agent of an attached proliferative agent, or or be empty. At each time step, each motile agent can attempt to move into an adjacent lattice site or proliferate at its current site. However, if a motile agent attempts to move into a site that is already occupied or reserved, the movement event is aborted. Similarly, if a motile agent attempts to begin proliferation by reserving a site that is already occupied, then the proliferation event is aborted. Agents attempt to convert from being motile to proliferative at constant rate $r$, and the proliferative phase has length $\tau$, upon which two motile daughter agents appear, one on the proliferating site, and one on the reserved site. 

Based on the above, tracking the rate of change of the number of motile agents, $m(t)$, proliferative agents, $p(t)$, and reserved sites, $q(t)$, in time, and following the arguments of \cite{Baker:2010:CMF}, we obtain the following mean-field approximation:
$$m'(t)=-r m(t)\frac{K\!-\!m(t)\!-\!p(t)\!-\!q(t)}{K}+2r m(t\!-\!\tau)\frac{K\!-\!m(t\!-\!\tau)\!-\!p(t-\tau)-q(t-\tau)}{K},$$
$$p'(t)=r m(t)\frac{K-m(t)-p(t)-q(t)}{K}-r m(t-\tau)\frac{K-m(t-\tau)-p(t-\tau)-q(t-\tau)}{K},$$
$$q'(t)=r m(t)\frac{K-m(t)-p(t)-q(t)}{K}-r m(t-\tau)\frac{K-m(t-\tau)-p(t-\tau)-q(t-\tau)}{K},$$
where the term $(K-m(t)-p(t)-q(t))/K$ expresses the probability that a randomly selected site is empty at time $t$.
Using the variable $u=K-m-p-q$ that accounts for empty sites, we can write

\begin{eqnarray}
m'(t)&=&-r K^{-1} m(t)u(t)+2rK^{-1} m(t-\tau)u(t-\tau) , \label{1} \\
p'(t)&=&r K^{-1} m(t)u(t)-rK^{-1} m(t-\tau)u(t-\tau), \label{2}\\
q'(t)&=&r K^{-1} m(t)u(t)-r K^{-1}m(t-\tau)u(t-\tau), \label{3}\\
u'(t)&=&- rK^{-1} m(t)u(t). \label{4}
\end{eqnarray}

\section{Long-term behaviour}

The usual phase space for Eqs. \eqref{1}-\eqref{4} is $C=C([-\tau,0],R^4)$, the Banach space of continuous function from the interval $[-\tau,0]$ to $R^4$ equipped with the supremum norm. With the notation $x(t)=(m(t),p(t),q(t),u(t))$, our system is of the form $x'(t)=f(x_t)$ where $x_t \in C$ is defined by the relation $x_t(\theta)=x(t+\theta)$ for $\theta \in [-\tau,0]$ and $f: C \to R^4$ is defined by the right-hand side of Eqs. \eqref{1}-\eqref{4}. The standard results for delay differential equations provide existence and uniqueness of solutions from initial data $x_0=\phi \in C$ (see, for example \cite{kuang}).

Given the biological motivation, we are interested only in non-negative solutions, for which $p(t)=q(t)=rK^{-1}\int_{t-\tau}^t m(s)u(s)\text{d}s $ holds, meaning that proliferative cells at a given time $t$ are exactly those who started the proliferation process in the time interval $[t-\tau,t]$, and the reserved sites correspond to them. With this compatibility condition and the balance law $K=m(t)+p(t)+q(t)+u(t)$, we define the feasible phase space
\begin{eqnarray}\Omega&:=& \left\{\vphantom{\sum_{j=1}^{4} \phi_j(0)}\phi \in C : \phi_j(\theta)\geq 0 \text{ for all } \theta \in [-\tau,0], j=1,2,3,4 ; \right. \nonumber \\ 
&& \qquad \left.\sum_{j=1}^{4} \phi_j(0)=K; \quad \phi_2(0)=\phi_3(0)=rK^{-1}\int_{-\tau}^0 \phi_1(s)\phi_4(s)\text{d}s  \right\}. \label{feasible}
\end{eqnarray}

\begin{lemma} The set $\Omega$ is forward invariant, that is for any solution $x(t)$ with $x_0 \in \Omega$, $x_t \in \Omega$ for all $t\geq 0$.
	\end{lemma}
\begin{proof}
Integrate Eq. \eqref{2} from $0$ to $t$ to obtain (similarly for $q(t)$)
$$p(t)-p(0)=rK^{-1}\int_{t-\tau}^t m(s)u(s)\text{d}s-rK^{-1}\int_{-\tau}^0 m(s)u(s)\text{d}s.$$
 From $x_0 \in \Omega$ we have
$p(0)=q(0)=rK^{-1}\int_{-\tau}^0 m(s)u(s)\text{d}s,$
hence 
\begin{equation}
p(t)=q(t)=rK^{-1}\int_{t-\tau}^t m(s)u(s)\text{d}s, \label{pqint}
\end{equation}
thus the third condition in the definiton of $\Omega$ is preserved. The second trivially follows from summing up the equations to see $(m(t)+p(t)+q(t)+u(t))'=0,$
so $K=m(t)+p(t)+q(t)+u(t)$ is preserved. To confirm nonnegativity, note that $u(t)=u(0)\exp(-rK^{-1}\int_{0}^t m(s)\text{d}s)\geq 0.$ Assuming that $m(t)\geq 0$ for $t\leq t_0$, we have \\ $m(t-\tau)u(t-\tau)\geq 0$ for $t \leq t_0+\tau$, and consequently $m(t)\geq m(t_0)\exp(-rK^{-1}\int_{t_0}^t u(s)\text{d}s)$ holds on $[t_0,t_0+\tau]$. Hence, by the method of steps we obtain non-negativity of $m(t)$ for all $t$. Then the non-negativity of $p(t)$ and $q(t)$ follow from Eq. \eqref{pqint}.	
\qed	\end{proof}

Note that since solutions starting from $\Omega$ stay in this bounded set, they exist globally. Following \cite{kuang}, we say that a continuous functional $V : C \to R$ is a Lyapunov functional on the set $\Omega$
in $C$ for Eqs. \eqref{1}-\eqref{4}, if it is continuous on the closure of $\Omega$, and $\dot{V} \leq 0$ on $\Omega$. Here, $\dot{V}$ denotes the derivative of $V$ along solutions. In our case $\Omega$ is itself closed. We also define
$E := \{\phi \in \Omega : \dot{V} = 0\}$ and
$M := \text{the largest set in $E$ which is invariant with respect to
Eqs. \eqref{1}-\eqref{4})}.$

\begin{theorem} If $m(0)>0$, then $\lim_{t\to\infty} (m(t),p(t),q(t),u(t))=(K,0,0,0)$.
\end{theorem}
\begin{proof}
Consider the functional $V(\phi)=\phi_4(0)$. Then $\dot{V}=-rK^{-1} m(t)u(t) \leq 0$ for solutions in $\Omega$, and by LaSalle's invariance principle (cf. Thm. 2.5.3 in \cite{kuang}), the limit set of any solution is in $M$, thus on the limit set of any solution, $m u\equiv 0$ holds. Since for any solution $u$ is always zero or always positive, we have either $m\equiv 0$ or $u\equiv 0$. In both cases, $p=q\equiv 0$ follows. Hence, the limit set can only be composed of the two equilibria $(K,0,0,0)$ or $(0,0,0,K)$. Finally, we show that if $m(0)>0$, then $m(t)$ can not converge to $0$. Since $u(t)$ is monotone decreasing, for such a solution $m(t)+p(t)+q(t)=K-u(t)\geq K-u(0)>0$ should hold. If $m(t)\to 0$ as $t \to \infty$, then from Eq. \eqref{pqint}, also $p(t)=q(t) \to 0$. This contradicts $m(t)+p(t)+q(t)\geq K-u(0)>0$ and so we can exclude $(0,0,0,K)$ from the limit set. Therefore  $\lim_{t\to\infty} (m(t),p(t),q(t),u(t))=(K,0,0,0)$.
\qed	\end{proof}
Remark: if $m(0)=0$, then also $p(0)=q(0)=0$, so $u(0)=K$ and we are on the empty lattice having the trivial solution $(0,0,0,K)$.

\begin{figure}[ht]
	\centering
	\includegraphics[scale=0.45]{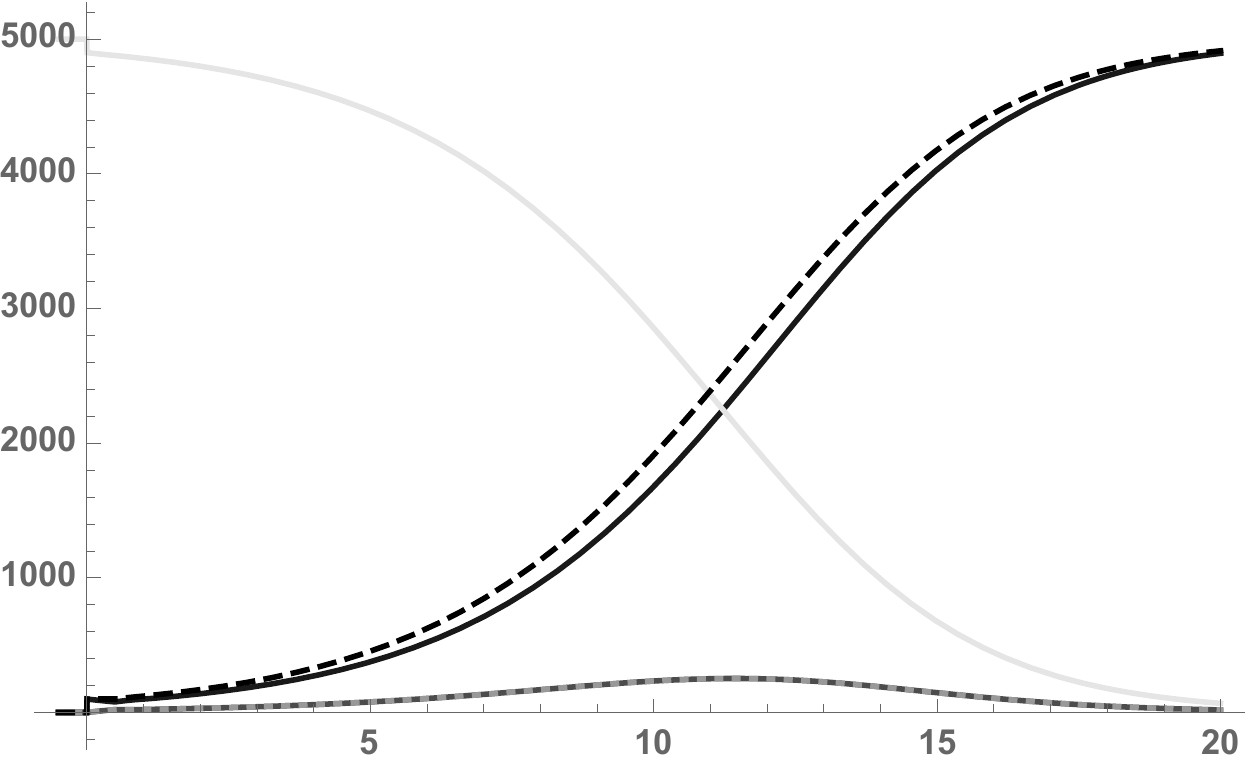}
	\includegraphics[scale=0.45]{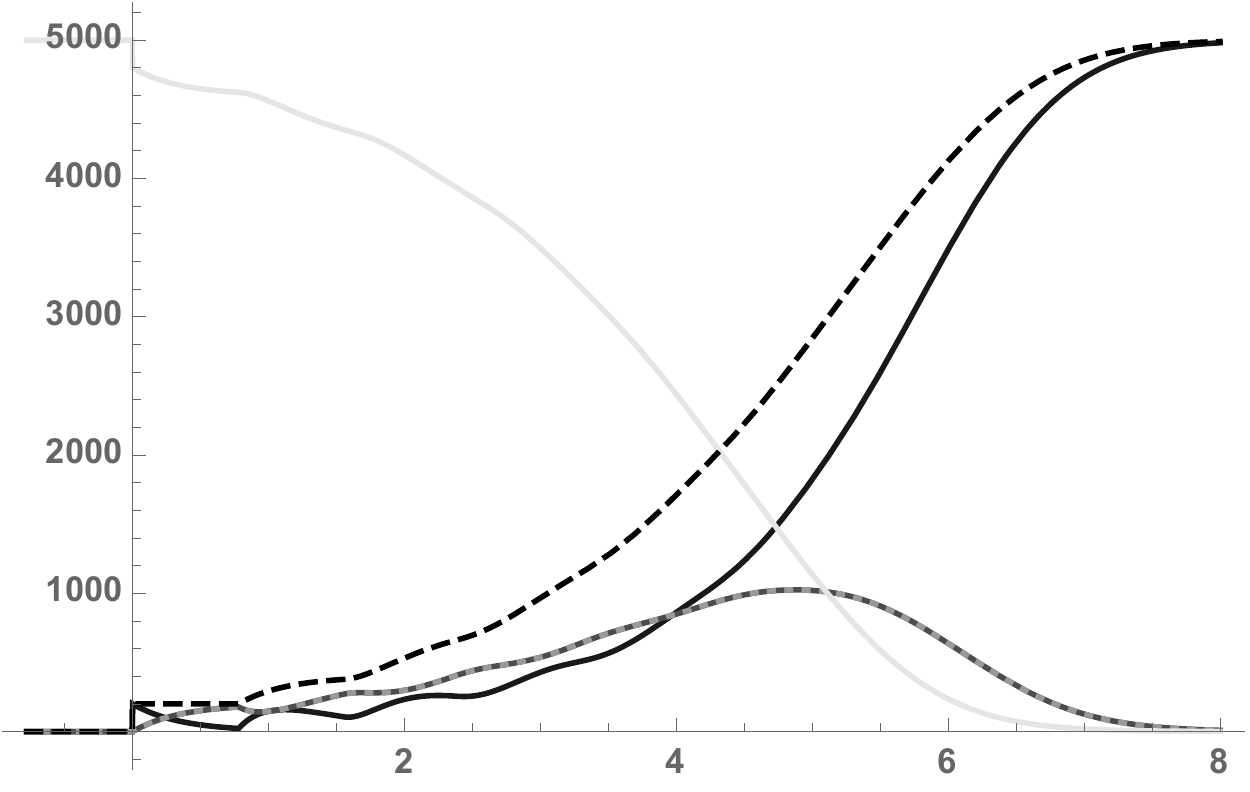}
	\includegraphics[scale=0.45]{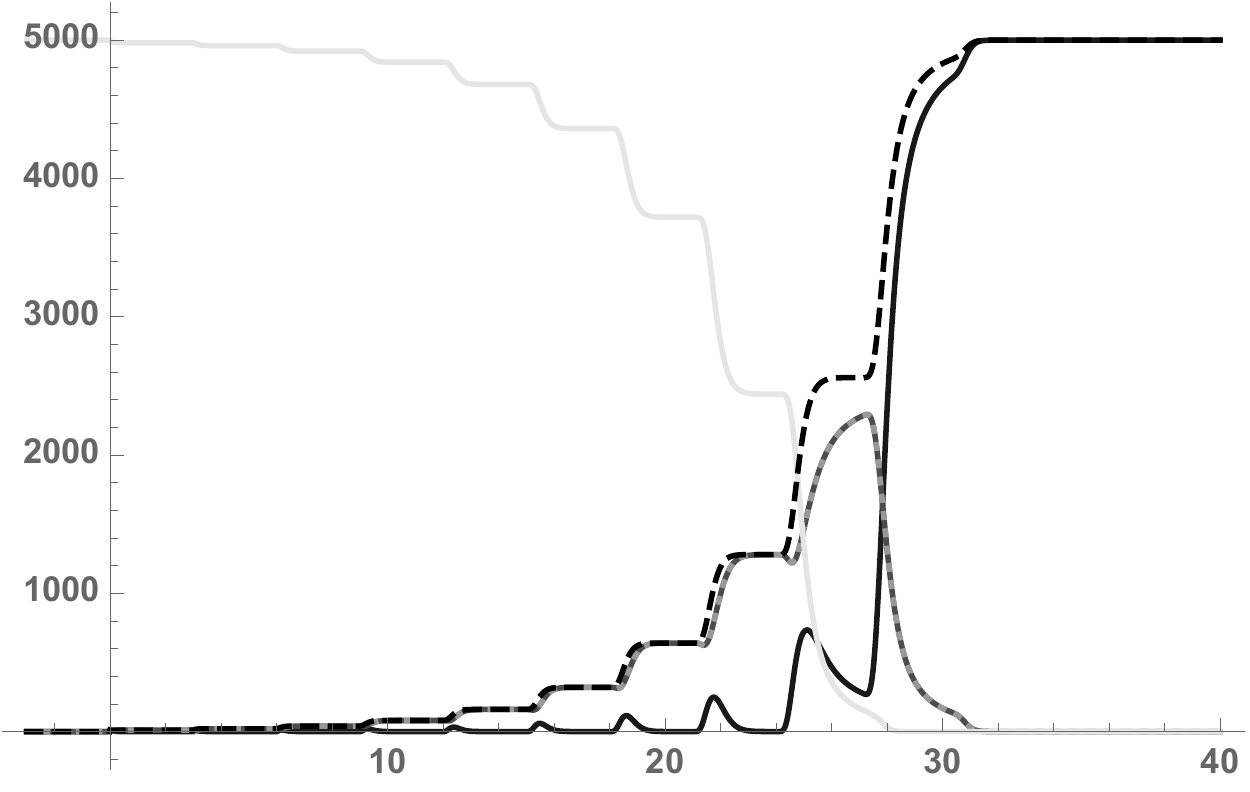}
	\includegraphics[scale=0.45]{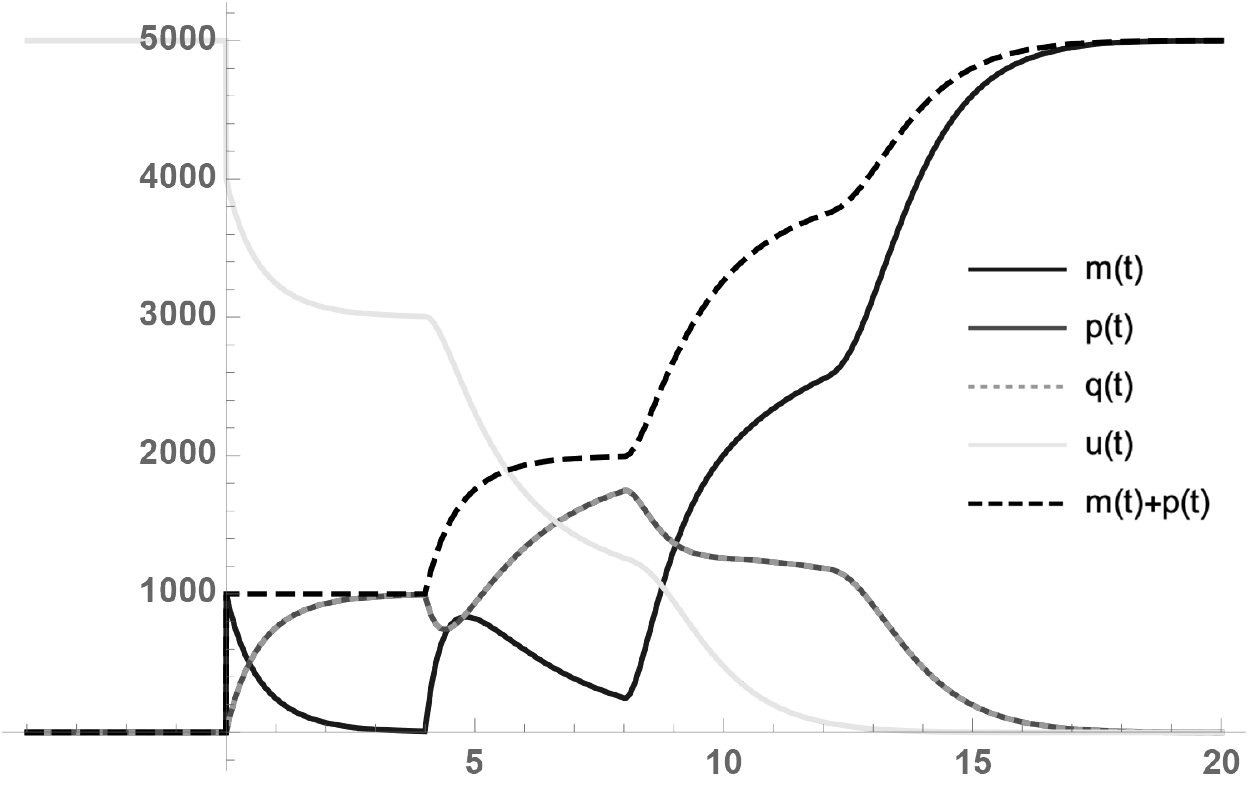}
	\caption{Four numerical simulations, where in each realization, the  initial function $\phi_1$ on [$-\tau,0$] is given by $a H(\theta)$, where $a>0$ is the initial number of cells,  $\phi_2 = \phi_3 = 0$ and $\phi_4 = K-a H(\theta)$. This choice of the initial function models an \textit{in vitro} experiment where motile cells are added to the plate at $t=0$. The parameters are the following: \textit{Top Left} -- $r=0.5,a=100,\tau=0.5$;	\textit{Top Right} -- $r=3,a=200,\tau=0.8$; \textit{Bottom Left} -- $r=10.5,a=10,\tau=3$; \textit{Bottom Right} -- $r=2,a=1000,\tau=4$; and $K=5000$ in each case. The legend in the bottom right figure applies to each.}
	\label{fig:Exp_1}
\end{figure}
\begin{figure}[ht]
	\centering
	\includegraphics[scale=0.45]{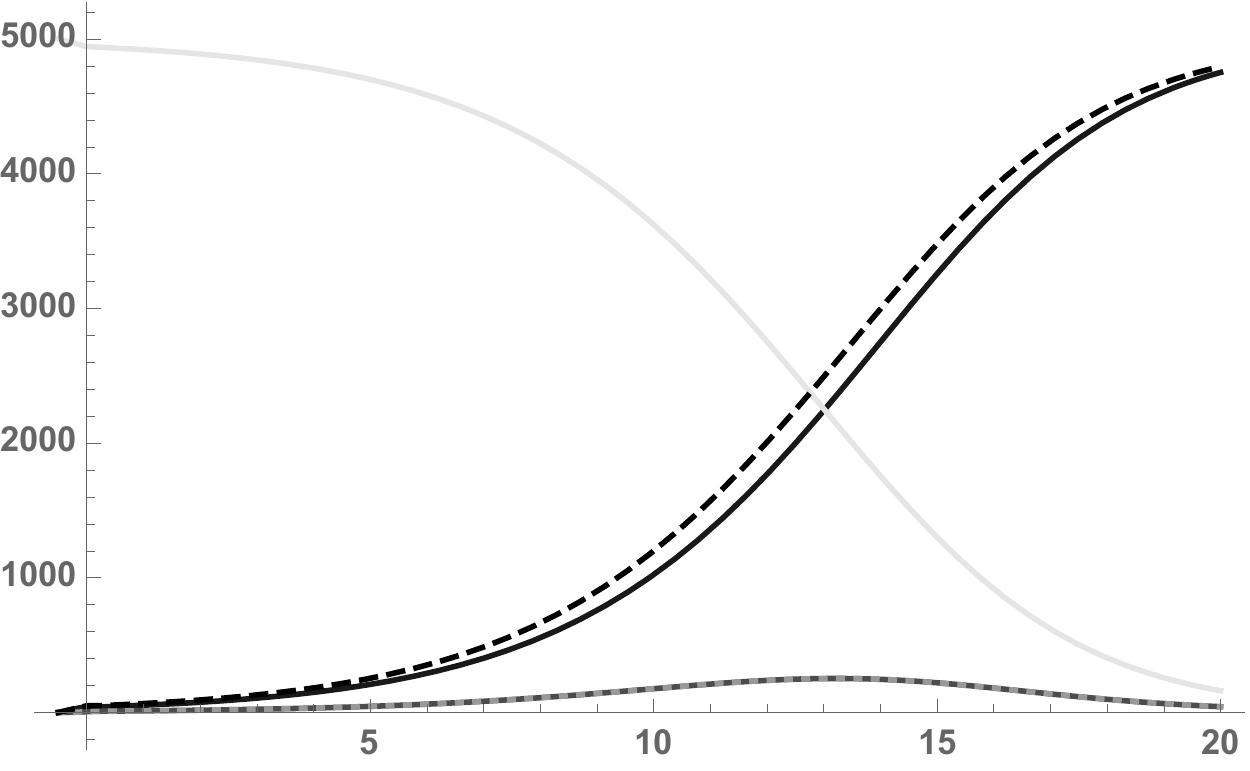}
	\includegraphics[scale=0.45]{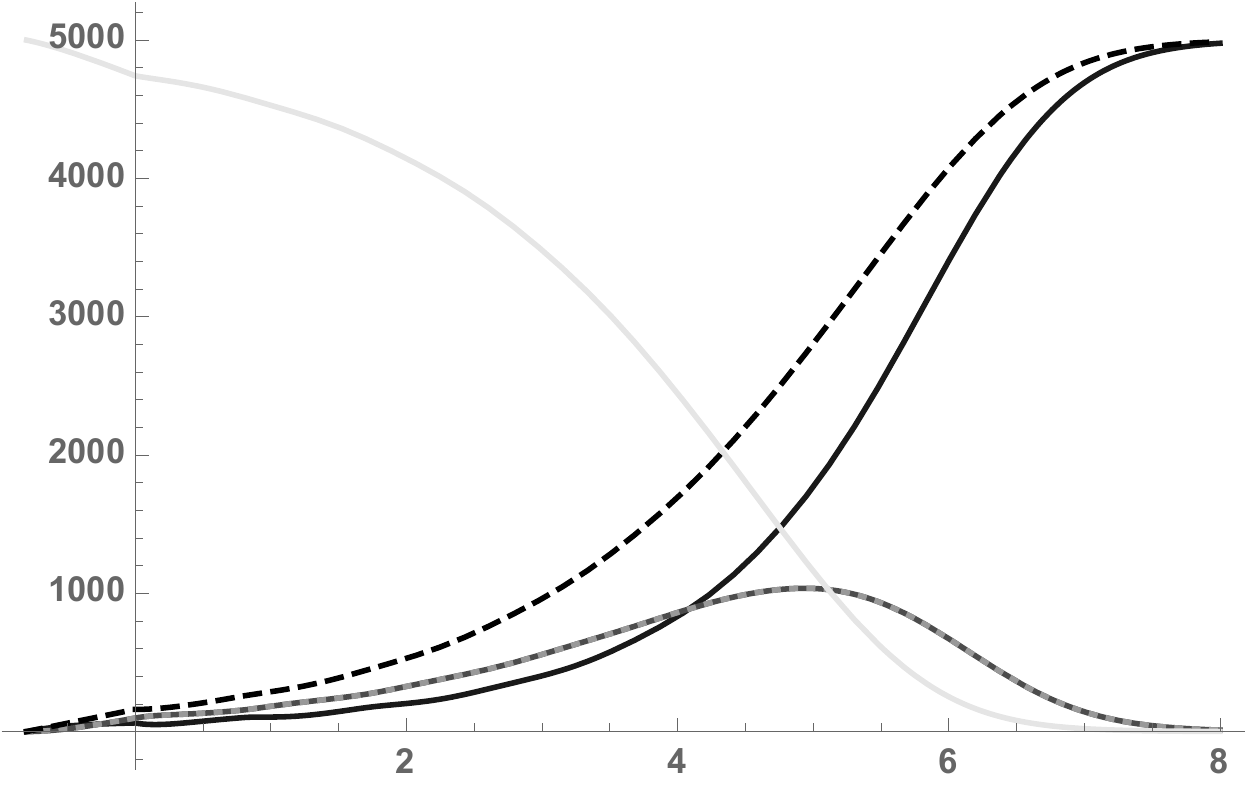}
	\includegraphics[scale=0.45]{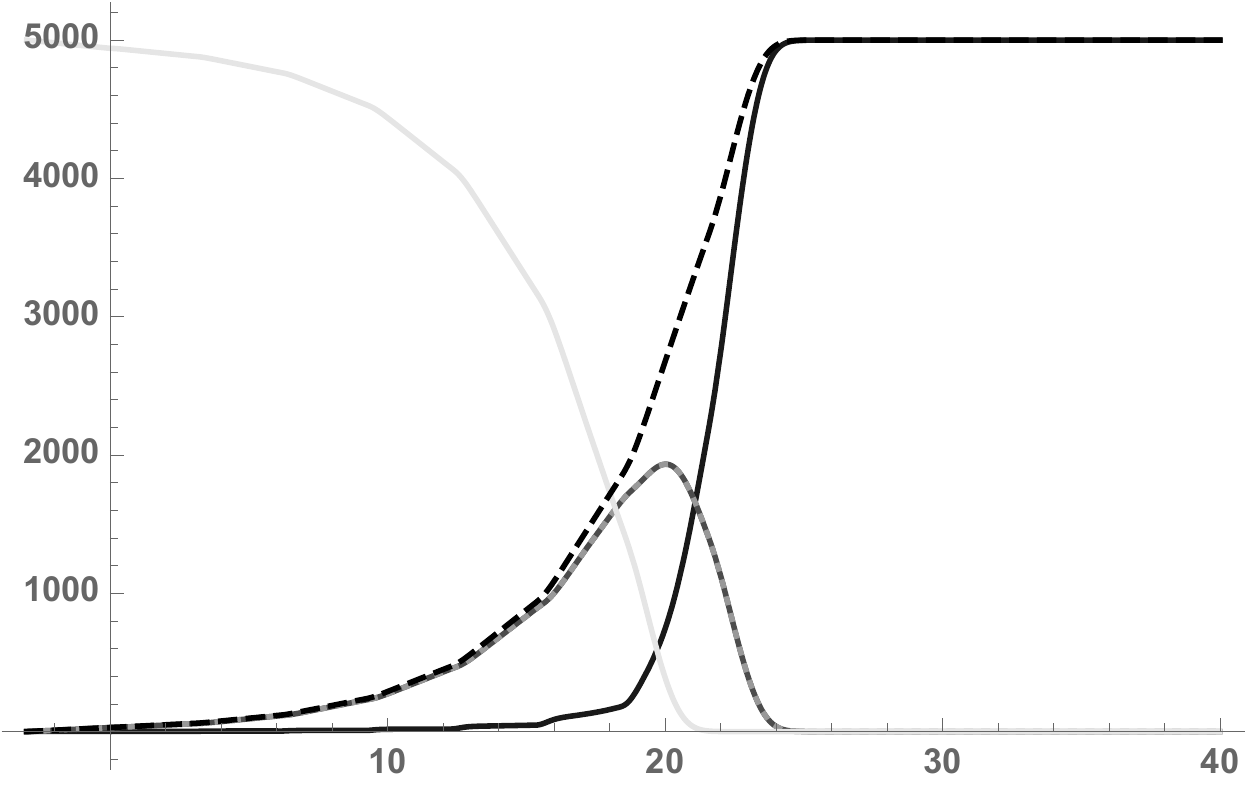}
	\includegraphics[scale=0.45]{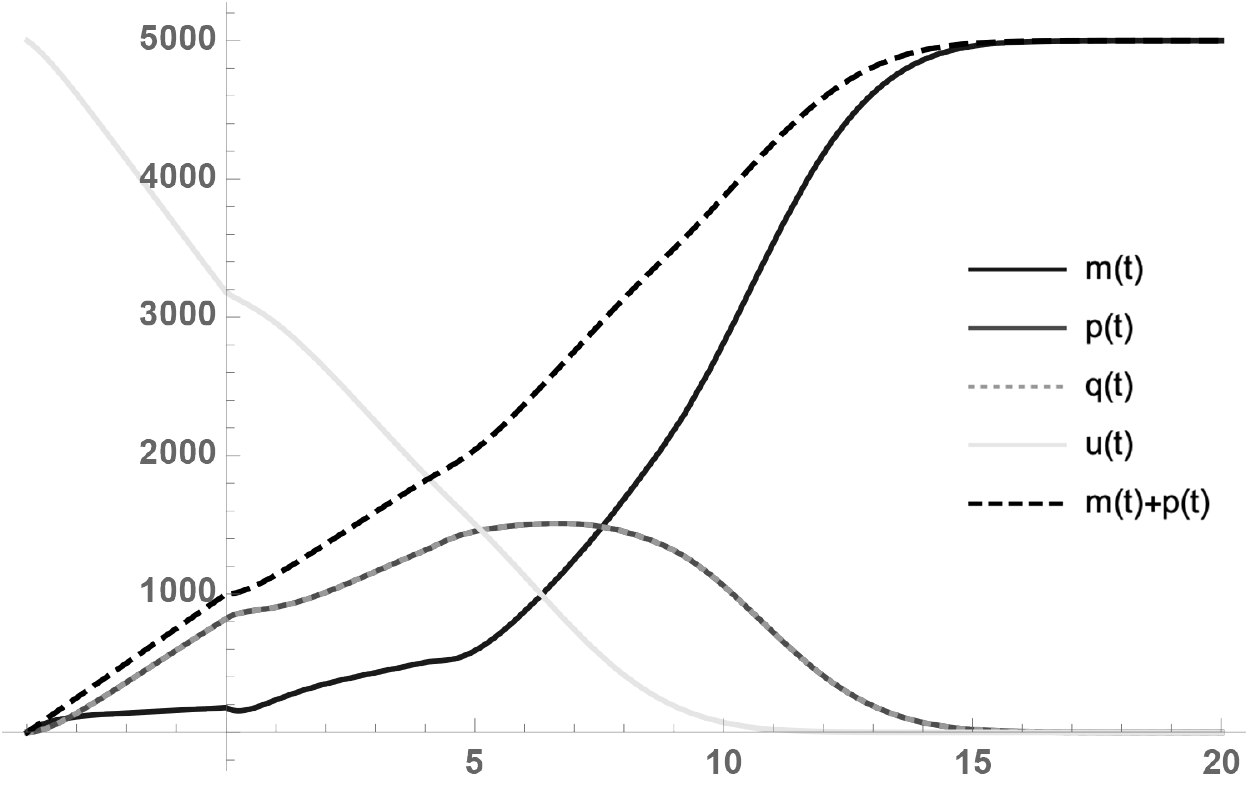}
	\caption{Four simulations, where in each realization motile cells are added with rate $a$ in the initial interval. The parameters are the following: \textit{Top Left} -- $r=0.5,a=100,\tau=0.5$;	\textit{Top Right} -- $r=3,a=200,\tau=0.8$; \textit{Bottom Left} -- $r=10.5,a=10,\tau=3$;	\textit{Bottom Right} -- $r=2,a=250,\tau=4$ and $K=5000$ in each case. The legend in the bottom right figure applies to each.}
	\label{fig:ExpCont}
\end{figure}

\section{Simulations and conclusion}

According to the choice of the initial functions, different \textit{in vitro} experiments can be modelled with Eqs. \eqref{1}-\eqref{4}. One approach is to add a number of motile cells all at once at $t=0$ to the empty cell space (e.g. a Petri dish). In this experiment the initial function $\phi_1$  is given by $\phi_1(\theta)= aH(\theta)$ for $\theta \in [-\tau,0]$, where $a$ stands for the  number of introduced cells at $t=0$, and $H(\theta)$ is the right-continuous Heaviside-function, i.e. $H(\theta)=0$ for $\theta<0$ and $H(\theta)=1$ for $\theta\geq 0$. In this setting, we take $\phi_2 (\theta)= \phi_3(\theta) = 0$ and $\phi_4 (\theta)= K-a H (\theta)$. While such initial data is not from $C$, they satisfy Eq. \eqref{feasible} and generate a continuous solution for $t>0$. Some of such simulations are shown in Figure \ref{fig:Exp_1}.

A more elaborate \textit{in vitro} experiment is the following. Instead of motile cells all at once, we add them in to the assay with a constant rate $a$ for a time interval of length $\tau$. After this, we leave the cell population intact. The initial data corresponding to this experiment can be obtained by solving a modification of Eqs. \eqref{1}-\eqref{4} with an additive forcing term $+a$ to the $m$-equation (and $-a$ to the $u$-equation), representing the gradual addition of $m$-cells, on an interval of length $\tau$, starting from the state $(0,0,0,K)$. Then we start solutions of Eqs. \eqref{1}-\eqref{4} with such initial functions, which satisfy Eq. \eqref{feasible}. Four realizations of this experimental setting are shown in Figure \ref{fig:ExpCont}.

The point of considering these two setups is that in the first we have only motile cells at $t=0$, while in the second at $t=0$ we have a distribution of cells in different phases of the cell cycle. This has a profound impact on the behaviour of solutions. While in Section 3 we proved that all solutions settle eventually at the state $(K,0,0,0)$, there are distinctive features of solutions in different scenarios. Figure 1 shows that when the cell cycle delay is small, the solutions resemble logistic growth. In contrast, when the delay is large relative to the average time between individual cells attempting enter the proliferative state, the initially motile cells enter the proliferative state more or less together, and hence complete  cell division more or less together too, resulting in a step-function-style growth curve in the total cell count. The sudden switching between phenotypes causes non-monotonic behaviours in $m(t)$ and $p(t)$ also. When we add motile cells continuously rather than adding them all at once, the solutions are much more similar to the expected logistic growth curve, and a different characteristic can be observed only for high proliferation rates or large numbers of initially added cells. In conclusion, an intermittent growth of a cell population can be an indication that the cell cycle length is relatively large (relative to inter-proliferation times), while its variance is small.

\medskip

\textbf{Acknowledgment} REB is a Royal
Society Wolfson Research Merit Award holder and would like to thank the Leverhulme
Trust for a Research Fellowship. GR was supported by Marie Sk\l odowska-Curie
Grant No. 748193. PB was supported by NKFI FK 124016 and EFOP-3.6.1-16-2016-00008.

\enlargethispage{0.5cm}

\end{document}